\documentclass[journal, 10pt]{IEEEtran}
\usepackage[utf8]{inputenc}
\usepackage{cite,color}
\usepackage{physics}
\usepackage{subcaption, graphicx, esvect}
\usepackage{amsmath,amssymb,amsthm}
\usepackage{url}
\usepackage[ruled,vlined]{algorithm2e}

\newtheorem{defn}{Definition}
\newtheorem{thrm}{Theorem}
\newtheorem{cor}{Corollary}
\newtheorem{lma}{Lemma}

\newtheorem{remark}{Remark}

\title{Strategic Mitigation of Agent Inattention in Drivers with Open-Quantum Cognition Models}

\author{\IEEEauthorblockN{Qizi Zhang\IEEEauthorrefmark{1}, Venkata Sriram Siddhardh Nadendla\IEEEauthorrefmark{2}, S. N. Balakrishnan\IEEEauthorrefmark{1}, Jerome Busemeyer\IEEEauthorrefmark{3}}
\\[1ex]
\IEEEauthorblockA{\IEEEauthorrefmark{1}Dept. of Mechanical and Aerospace Engineering and \IEEEauthorrefmark{2}Dept. of Computer Science \\
Missouri Univ. of Science and Technology \\ Rolla, Missouri 65401. \\ Email: \{qzwtb, nadendla, bala\}@mst.edu
}
\\[1ex]
\IEEEauthorblockA{\IEEEauthorrefmark{3}Dept. of Psychological and Brain Sciences, \\ Indiana University Bloomington \\
Bloomington, IN 47405. \\
Email: jbusemey@indiana.edu}
\vspace{-5ex}
}

\begin{document}

\maketitle

\begin{abstract}
State-of-the-art driver-assist systems have failed to effectively mitigate driver inattention and had minimal impacts on the ever-growing number of road mishaps (e.g. life loss, physical injuries due to accidents caused by various factors that lead to driver inattention). This is because traditional human-machine interaction settings are modeled in classical and behavioral game-theoretic domains which are technically appropriate to characterize strategic interaction between either two utility maximizing agents, or human decision makers. Therefore, in an attempt to improve the persuasive effectiveness of driver-assist systems, we develop a novel strategic and personalized driver-assist system which adapts to the driver's mental state and choice behavior. First, we 
propose a novel equilibrium notion in human-system interaction games, where the system maximizes its expected utility and human decisions can be characterized using any general decision model. Then we use this novel equilibrium notion to investigate the strategic driver-vehicle interaction game where the car presents a persuasive recommendation to steer the driver towards safer driving decisions. We assume that the driver employs an open-quantum system cognition model, which captures complex aspects of human decision making such as violations to classical law of total probability and incompatibility of certain mental representations of information. We present closed-form expressions for players' final responses to each other's strategies so that we can numerically compute both pure and mixed equilibria. Numerical results are presented to illustrate both kinds of equilibria.
\end{abstract}

\section{Introduction}

Driver inattention is a dangerous phenomenon that can arise because of various reasons: distractions, drowsiness due to fatigue, less reaction time due to speeding, and intoxication. The consequences of inattentive driving can severely affect the driver's safety even under normal road conditions, and can be devastating in terms of life-loss and/or long-lasting injuries. According to NHTSA's latest revelations \cite{NHTSA}, in 2019, 3142 lives were claimed by distracted driving, 795 lives were claimed by drowsy driving, 9378 deaths were due to speeding, and 10,142 deaths were due to drunk driving, all in the United States alone. Therefore, several types of driver-assist systems have been developed and deployed in modern vehicles to mitigate inattentiveness. However, traditional driver-assist technologies are static and not personalized, which are insufficient to handle the situations in futuristic transportation systems with mostly connected and/or autonomous vehicles. For example, several deadly accidents have been reported where the Tesla driving assistants were working normally but the drivers were inattentive \cite{Tesla, Tesla2}. As per SAE standard J3016 \cite{SAE}, the state-of-the-art vehicles mostly fall under Levels 2/3, which continue to demand significant driver attention (e.g. Tesla autopilot \cite{Tesla3}), especially in uncertain road and weather conditions.  Therefore, there is a strong need to design dynamic, data-driven driver-alert systems which present effective interventions in a strategic manner based on its estimates of the driver's attention level and physical conditions. 

However, the design of strategic interventions to mitigate the ill effects of driver inattention is quite challenging due to three fundamental reasons. 
Firstly, the \emph{driver may not follow the vehicle's recommendations} (i) if the driver is inattentive, (ii) if the driver does not trust the vehicle's recommendations, and/or (iii) if the recommendation signal is not accurate enough to steer driver's choices (e.g. the driver may not stop the vehicle because of a false alarm). Secondly, the \emph{persuasive effectiveness of vehicle's recommendations is technically difficult to evaluate} due to its complex/unknown relationship with the driver's (i) attention level \cite{koopman2017autonomous}, (ii) own judgment/prior of road conditions \cite{woide2019methodical}, and (iii) trust on the vehicle's recommendation system \cite{choi2015investigating}. In addition, it is difficult to mathematically model and estimate these three terms \cite{nishigaki2019driver,zyner2017long,abe2006alarm}. Finally, there is strong evidence within the psychology literature that \emph{human decisions exhibit several anomalies to traditional decision theory}. Examples include deviations from expected utility maximization such as Allais paradox \cite{tversky1992advances}, Ellsberg paradox \cite{ellsberg1961risk}, violations of transitivity and/or independence between alternatives \cite{busemeyer1993decision}; and deviations from classical Kolmogorov probability theory such as conjunction fallacy \cite{franco2009conjunction}, disjunction fallacy \cite{young2007potential}, and violation of sure thing principle \cite{khrennikov2009quantum}.

There have been a few relevant efforts in the recent literature where both the driver and the driver-assist system interact in a game theoretic setting. These efforts can be broadly classified into two types: (i) the \emph{direct method} where the system uses its on-board AI to directly control the vehicle, and (ii) the \emph{indirect method} where the system indirectly controls the vehicle via relying on the driver to make decisions. On the one hand, Flad \emph{et al.} proposed a direct method that models driver steering motion as a sequence of motion primitives so that the aims and steering actions of the driver can be predicted and then the optimal torque can be calculated \cite{7929390}. Another example that proposes a direct method is by Na and Cole in \cite{na2014game}, where four different paradigms were investigated: (i) decentralized, (ii) non-cooperative Nash, (iii) non-cooperative Stackelberg, and (iv) cooperative Pareto, to determine the most effective method to model driver reactions in collision avoidance systems. 
Although direct methods can mimic driver actions, they certainly do not consider the driver's cognition state (in terms of preferences, biases and attention) and no intervention was designed/implemented to mitigate inattention. On the other hand, indirect methods have bridged this gap via considering driver's cognition state into account. Lutes \emph{et al.} modeled driver-vehicle interaction as a Bayesian Stackelberg game, where the on-board AI in the vehicle (leader) presents binary signals (no-alert/alert) based on which the driver (follower) makes a binary decision (continue/stop) regarding controlling the vehicle on a road \cite{lutes2020perfect}. This work and \cite{lutes2020perfect} share the same setting of unknown road condition and binary actions of two players, and also introduce a non-negative exponent parameter in the overall driver’s utility to capture his/her level of attention. The difference is that \cite{lutes2020perfect} still follows the traditional game theory framework of maximizing payoffs while this work extends the traditional framework in which the players do not necessarily maximize payoffs.  Schwarting \emph{et al.} integrated Social Value Orientation (SVO) into autonomous-vehicle decision making. Their model quantifies the degree of an agent’s selfishness or altruism in order to predict the social behavior of other drivers. They modeled interactions between agents as a best-response game wherein each agent negotiates to maximize their own utilities \cite{schwarting2019social}. However, all the human players in the game of the above research are still assumed to be rational players maximizing utilities, even though the utilities are modified to capture attention level or social behavior, whether by a non-negative exponent parameter or by SVO. The present work bridges this gap by directly considering the driver as an agent who does not seek to maximize payoff, but instead uses a quantum-cognition based decision process to make decisions.



Note that most of the past literature focused on addressing each of these challenges independently. The main contribution of this paper is that we address all the three challenges jointly in our driver-vehicle interaction setting. In Section \ref{section:formulation}, we propose a novel strategic driver-vehicle interaction framework where all the aforementioned challenges are simultaneously addressed in a novel game-theoretic setting. We assume that the vehicle constructs recommendations so as to balance a prescribed trade-off between information accuracy and persuasive effectiveness. On the other hand, we model driver decisions using an open quantum cognition model that considers driver attention as model parameter and incorporates the driver prior regarding road condition into the initial state. In Section \ref{section:lindblad}, we present a closed-form expression for the cognition matrix in the driver's open quantum cognition model. Given that the agent rationalities are fundamentally different from each other (vehicle being a utility-maximizer, and driver following an open quantum cognition model), we also propose a novel equilibrium notion, inspired by Nash equilibrium, and compute both pure and mixed equilibrium strategies for the proposed driver-vehicle interaction game in Sections \ref{section:pure} and \ref{section:mix} respectively. Finally, we analyze the impact of driver inattention on the equilibrium of the proposed game.

\section{Strategic Driver-Vehicle Interaction Model}
\label{section:formulation}

In this section, we model the strategic interaction between a driver-assist system (car) and an inattentive driver as a one-shot Bayesian game. We assuming that the physical conditions of the road are classified into two states, namely, \emph{safe} (denoted as $S$) and \emph{dangerous} (denoted as $D$). The vehicle can choose one of the two signaling strategies: alert the driver (denoted as $A$), or no-alert (denoted as $N$) based on its belief about the road state. Meanwhile, based on the driver's belief about the road state and his/her own mental state (which defines driver's type), the driver chooses to either \emph{continue} driving (denoted as $C$), or \emph{stop} the vehicle (denoted as $S$). Note that although the letter $S$ is used to denote both road state being safe and driver decision being stop, the reader can easily decipher the notation's true meaning from context.

Depending on the true road state, we assume that the vehicle (row player) and the driver (column player) obtain utilities as defined in Table \ref{tab:tab1}. When the road is dangerous, we expect the car to alert the driver. If the car does not alert, it will get a low payoff. Furthermore, we assume this low payoff depends on the driver's action. If the driver stops, the payoff is only slightly low because no damage or injury is incurred. If the driver continues to drive, the payoff is very low because damage or injury is incurred. When the road is safe, the correct action for the car is not to alert. If the car does not alert, it will get a high payoff. This high payoff depends on the driver's action. If the driver stops, the payoff is only slightly high because it does not help the driver and an unnecessary stop is waste of time and energy. If the driver continues to drive, the reward is very high because everything is fine.

\begin{table}[htbp]
\centering
\begin{tabular}{r l}
\begin{tabular}{r|c|c|}
\multicolumn{1}{r}{} &  \multicolumn{1}{c}{C} & \multicolumn{1}{c}{S} 
\\
\cline{2-3}
N & $a_{1,s}$, $a_{2,s}$ & $b_{1,s}$, $b_{2,s}$ 
\\
\cline{2-3}
A & $c_{1,s}$, $c_{2,s}$ & $d_{1,s}$, $d_{2,s}$ 
\\
\cline{2-3}
\end{tabular}
&
\begin{tabular}{ r|c|c| }
\multicolumn{1}{r}{}
 &  \multicolumn{1}{c}{C}
 & \multicolumn{1}{c}{S} \\
\cline{2-3}
N & $a_{1,d}$, $a_{2,d}$ & $b_{1,d}$, $b_{2,d}$ \\
\cline{2-3}
A & $c_{1,d}$, $c_{2,d}$ & $d_{1,d}$, $d_{2,d}$ \\
\cline{2-3}
\end{tabular}
\end{tabular}
\caption{Utilities of the car and the driver when the road is safe (left) and dangerous (right)}
\label{tab:tab1}
\end{table}



In this paper, we assume that both the car and the driver does not know the true road state. While the car relies on its observations from on-board sensors and other extrinsic information sources (e.g. nearby vehicles, road-side infrastructure) and its on-board machine learning algorithm for road judgment to construct its belief $q \in [0,1]$ regarding the road state being safe, we assume that the driver constructs a belief $p \in [0,1]$ regarding the road state being safe based on what he/she sees and his/her prior experiences.
Furthermore, as in the case of a traditional decision-theoretic agent, we assume that the car seeks to maximize its expected payoff. If $p_C$ is the probability with which the driver chooses $C$, then the expected payoff for choosing $N$ and $A$ at the car are respectively given by
\begin{equation}
\begin{array}{lcl}
U_N(p_C) & = & p_C \Big[ a_{1, s} q + (1-q) a_{1, d} \Big] 
\\[1.5ex]
&& \qquad + \ (1- p_C) \Big[ b_{1, s} q + (1-q) b_{1, d} \Big]
\end{array}
\label{Eqn: Car's Utilities - N}
\end{equation}
and 
\begin{equation}
\begin{array}{lcl}
U_A(p_C) & = & p_C \Big[ c_{1, s}q + (1-q) c_{1, d} \Big] 
\\[1.5ex]
&& \qquad + \ (1- p_C) \Big[ d_{1, s}q + (1-q) d_{1, d} \Big].
\end{array}
\label{Eqn: Car's Utilities - A}
\end{equation}

The calculation of $p_C$ is complicated by the fact that the driver exhibits bounded rationality. Fortunately, the bounded rationality can be characterized by the open quantum system cognition model, as described below. 

\subsection{Driver's Open Quantum Cognition Model}
\label{subsec: oq}

In this subsection, we present the basic elements of the open quantum system cognition model \cite{martinez2016quantum}, and how it is applied to model driver behavior. The cognitive state of the agent is described by a mixed state or density matrix $\rho$, which is a statistical mixture of pure states. Formally it is a Hermitian, non-negative operator, whose trace is equal to one. Under the Markov assumption  (the evolution $\mathcal{E}$ can be factorized as $\mathcal{E}_{t_2,t_0}=\mathcal{E}_{t_2,t_1}\mathcal{E}_{t_1,t_0}$ given a sequence of instants $t_0$, $t_1$, $t_2$), one can find the most general form of this time evolution based on a time local master equation $d\rho/dt=\mathcal{L}[\rho]$, with $\mathcal{L}$  a differential superoperator (it acts over operators) called Lindbladian, which is defined as follows.   
\begin{defn}[\cite{rivas2012open}]
The Lindblad-Kossakowski equation for any open quantum system is defined as
\begin{equation}\label{eq:lindblad}
\begin{array}{l}
\displaystyle \frac{d\rho}{dt} = -i(1-\alpha)[H,\rho] + \displaystyle \alpha\sum_{m,n} \gamma_{m,n} \Big[ L_{m,n}\rho L^\dagger_{m,n} 
\\[1ex]
\displaystyle \qquad \qquad \qquad \qquad \qquad \qquad -\frac{1}{2}\left\{L^\dagger_{m,n}L_{m,n}, \ \rho \right\} \Big],
\end{array}
\end{equation}
where
\begin{itemize}
\setlength{\itemsep}{1ex}
\item $H$ is the Hamiltonian of the system,
\item $[H,\rho] = H\rho-\rho H $ is the commutation operation between the Hamiltonian $H$ and the density operator $\rho$,
\item $\gamma_{m,n}$ are $(m,n)^{th}$ entry of some positive semidefinite matrix (denoted as $C$),
\item $L_{m,n}$ is a set of linear operators,
\item $\left\{ L^\dagger_{m,n}L_{m,n}, \ \rho \right\} = L^\dagger_{m,n}L_{m,n}\rho+\rho L^\dagger_{m,n}L_{m,n}$ denotes the anticommutator. The superscript $\dagger$ represents the adjoint (transpose and complex conjugate) operation. 
\end{itemize}
\label{Defn: Lindblad-Kossakowski}
\end{defn}

In this paper, we set 
\begin{equation}
L_{(m,n)}=\ket{m}\bra{n}
\end{equation}
as defined in \cite{martinez2016quantum}, where, for any $m$, $\ket{m}$ is a column vector whose $m$th entry is 1 and the other entries are 0. Note that $\bra{n}$ is obtained by transposing $\ket{n}$ and then taking its complex conjugate. Thus $\bra{n}$ is a row vector whose $n$th entry is 1 and 0 otherwise.

The second term on the right side of Equation \eqref{eq:lindblad} contains the dissipative term responsible for the irreversibility in the decision-making process \cite{martinez2016quantum}, weighted by the coefficient $\alpha$ such that the parameter $\alpha\in [0,1]$ interpolates between the von Neumann evolution $(\alpha=0)$ and the completely dissipative dynamics $(\alpha=1)$. Furthermore, the term $\gamma_{(m,n)}$ is the $(m,n)$-th entry in the cognitive matrix $C(\lambda,\phi)$. This cognitive matrix $C(\lambda,\phi)$ is formalized as the linear combination of two matrices $\Pi(\lambda)$ and $B$, which are associated to the profitability comparison between alternatives and the formation of beliefs, respectively \cite{martinez2016quantum}: 
\begin{equation}\label{eq:cognition matrix}
\begin{array}{lcl}
C(\lambda,\phi) & = & \begin{bmatrix} 
\gamma_{(1,1)} & \cdots & \gamma_{(1,N)}
\\
\vdots & \ddots & \vdots
\\
\gamma_{(N,1)} & \cdots & \gamma_{(N,N)}
\end{bmatrix}
\\[6ex]
& = & (1-\phi) \cdot \Pi^T(\lambda) \ + \ \phi \cdot B^T,
\end{array}
\end{equation}
where $\phi\in[0,1]$ is a parameter assessing the relevance of the formation of beliefs during the decision-making process, $\Pi(\lambda)$ is the transition matrix where $(i,j)$-th entry $\pi_{ij}(\omega_l)$ is the probability that the decision maker switches from strategy $s_i$ to $s_j$ for a given state of the world $\omega_l$, and $B$ matrix allows the driver to introduce a change of belief about the state of the world in the cognitive process by jumping from one connected component associated to a particular state of the world $\omega_k\in\Omega$ to the connected component associated to another one $\omega_l\in\Omega$, while keeping the action $s_i$ fixed. The superscript $T$ denotes the transpose matrix. Finally, the dimension of the square matrix $C(\lambda,\phi)$, i.e. $N$, can be inferred from the detailed discussion given below.

\begin{figure*}[!t]
\centering
\begin{equation}
\begin{array}{lcl}
\Pi(\lambda) 
& = & \quad \begin{bmatrix} 
\mu(\lambda) & 1-\mu(\lambda) & 0 & 0 & 0 & 0 & 0 & 0 \\
\mu(\lambda) & 1-\mu(\lambda) & 0 & 0 & 0 & 0 & 0 & 0 \\
 0 & 0 & \nu(\lambda) & 1-\nu(\lambda) & 0 & 0 & 0 & 0 \\  0 & 0  & \nu(\lambda) & 1-\nu(\lambda) & 0 & 0 & 0 & 0 \\  0 & 0 & 0 & 0 & \xi(\lambda) & 1-\xi(\lambda) & 0 & 0 \\ 0 & 0 & 0 & 0 & \xi(\lambda) & 1-\xi(\lambda) & 0 & 0 \\ 0 & 0 & 0 & 0 & 0 & 0 & o(\lambda) & 1-o(\lambda) \\ 0 & 0 & 0 & 0 & 0 & 0 & o(\lambda) & 1-o(\lambda)
\end{bmatrix}.
\end{array}
\label{eq:PI}
\end{equation}
\end{figure*}

At the driver, the world state primarily consists of two components: (i) the road condition, and (ii) the car's action, i.e., the set of world states of the driver is $\Omega=\{SN, SA, DN, DA\}$ where the first letter represents road condition and the second letter represents car action. The utilities of the driver for choosing a strategy at a world state are as follows:
\begin{equation*}
\begin{array}{lclclcl}
u(C|SN) & = & a_{2,s}, & \ & u(S|SN) & = & b_{2,s},
\\[1ex]
u(S|SA) & = & c_{2,s}, & \ &  u(S|SA) & = & d_{2,s},
\\[1ex]
u(C|DN) & = & a_{2,d}, & \ &  u(S|DN) & = & b_{2,d},
\\[1ex]
u(S|DA) & = & c_{2,d}, & \ &  u(S|DA) & = & d_{2,d}.
\end{array}
\end{equation*}

We choose the basis of the road-car-driver system spanning the space of states to be
\begin{equation}
\begin{array}{l}
\Big\{ \ket{e_1},\ket{e_2},\ket{e_3},\ket{e_4},\ket{e_5},\ket{e_6},\ket{e_7},\ket{e_8} \Big\}
\\[2ex]
= \Big\{ \ket{SNC}, \ \ket{SNS}, \ \ket{SAC}, \ \ket{SAS}, \ \ket{DNC}, \ \ket{DNS},
\\[1ex]
\qquad \qquad \qquad \ket{DAC}, \ \ket{DAS} \Big\}.
\end{array}
\label{Eqn: Basis States}
\end{equation}

Next we define the transition matrix $\Pi(\lambda)$. If the utility of the decision maker by choosing strategy $s_i$ at the world state of $\omega_l$ is $u(s_i|\omega_l)$, the transition probability that the decision maker would switch to strategy $s_i$ at time step $k+1$ from strategy $s_j$ at time step $k$ is given in the spirit of Luce’s choice axiom \cite{luce1977choice,luce2012individual,yellott1977relationship}:
\begin{equation}
\label{eq:luce}
\pi_{(s_j \rightarrow s_i |\omega_l)} = P(s_i|s_j,\omega_l)=\frac{u(s_i|\omega_l)^\lambda}{\displaystyle \sum_{j=1}^{N_S}u(s_j|\omega_l)^\lambda}, 
\end{equation}
where the exponent $\lambda\geq 0$ measures the decision maker’s ability to discriminate the profitability among the different options. When $\lambda=0$, each strategy $s_i\in S$  has the same probability of being chosen ($1/N_S$), and when $\lambda\rightarrow\infty$ only the dominant alternative is chosen. There are two implications in this formulation of $P(s_i|s_j,\omega_l)$: (1) $u(s_i|\omega_l)\geq 0$ to avoid negative $P(s_i|s_j,\omega_l)$; (2) $P(s_i|s_j,\omega_l)$ only depends on the destination $s_i$ and does not depend on the starting point $s_j$.

Below are the probabilities needed for the $\Pi$ matrix:
\begin{equation*}
\begin{array}{lr}
\mu(\lambda) = \displaystyle \frac{a^\lambda_{2,s}}{a^\lambda_{2,s}+b^\lambda_{2,s}}, & 
\nu(\lambda) = \displaystyle \frac{c^\lambda_{2,s}}{c^\lambda_{2,s}+d^\lambda_{2,s}}, 
\\[2ex]
\xi(\lambda) = \displaystyle \frac{a^\lambda_{2,d}}{a^\lambda_{2,d}+b^\lambda_{2,d}}, & 
o(\lambda) = \displaystyle \frac{c^\lambda_{2,d}}{c^\lambda_{2,d}+d^\lambda_{2,d}},
\end{array}
\end{equation*}
where
\begin{itemize}
\setlength{\itemsep}{1ex}
\item $\mu(\lambda)$ is the probability that driver picks $C$ when he/she assumes that road state is $S$ and the car chooses $N$,
\item $\nu(\lambda)$ is the probability that driver will pick $C$ when he/she assumes that road state is $S$ and the car chooses $A$,
\item $\xi(\lambda)$ is the probability that driver will pick $C$ when he/she assumes that road state is $D$ and the car chooses $N$,
\item $o(\lambda)$ is the probability that driver will pick $C$ when he/she assumes that road state is $D$ and the car chooses $A$.
\end{itemize}
Equation \eqref{eq:PI} puts all the terms together in a matrix form and demonstrates the physical meaning of the row and column labels in $\Pi(\lambda)$. 

The $H$ matrix in Equation \eqref{eq:lindblad} is set as in \cite{martinez2016quantum}. When the elements of the $\Pi(\lambda)$ is nonzero, the elements of $H$ in the same position is 1; Otherwise it is zero. Thus, the $H$ matrix is 
\begin{equation}\label{eq:H}
\begin{array}{lcl}
H & = &
\begin{bmatrix} 
1 & 1 & 0 & 0 & 0 & 0 & 0 & 0 \\
1 & 1 & 0 & 0 & 0 & 0 & 0 & 0 \\
 0 & 0 & 1 & 1 & 0 & 0 & 0 & 0 \\  0 & 0  & 1 & 1 & 0 & 0 & 0 & 0 \\  0 & 0 & 0 & 0 & 1 & 1 & 0 & 0 \\ 0 & 0 & 0 & 0 & 1 & 1 & 0 & 0 \\ 0 & 0 & 0 & 0 & 0 & 0 & 1 & 1 \\ 0 & 0 & 0 & 0 & 0 & 0 & 1 & 1
\end{bmatrix}
\\[11ex]
& = & \begin{bmatrix} 
1 & 1 \\1 & 1\end{bmatrix}\oplus\begin{bmatrix} 
1 & 1 \\1 & 1\end{bmatrix}\oplus\begin{bmatrix} 
1 & 1 \\1 & 1\end{bmatrix}\oplus\begin{bmatrix} 
1 & 1 \\1 & 1\end{bmatrix}.
\end{array}
\end{equation}


In this paper, we set $\phi=0$ for the following two reasons: (1) Since the world state of the driver is mainly the action of the car and the action of the car is known when calculating the equilibrium, the driver does not need to form such a belief; (2) We are considering a one-shot game and we can assume the road condition does not change in one game, i.e., we are only considering short-time dynamic. The $B$ matrix is zeroed out and its content is not described here. Thus $C=\Pi$ and we set $\gamma_{m,n}=C_{m,n}$ in Equation \eqref{eq:lindblad}.

\subsection{Pure and Mixed Strategy Equilibria}
\label{section:eq}

For the sake of simplicity, let us denote the car as Agent 1, and the driver as Agent 2 without any loss of generality. Since the car seeks to maximize its expected payoff given that the driver chooses a strategy $s_2 \in \{C,S\}$, it is natural that the car's final response $FR_1(s_2)$ is its best response that maximizes its expected payoff given in Equations \eqref{Eqn: Car's Utilities - N} and \eqref{Eqn: Car's Utilities - A}, i.e., 
$$FR_1(s_2) = BR_1(s_2) \left( \triangleq \displaystyle \max_{s_1 \in \{ N,A \} } U_{s_1}(s_2) \right).$$ 

On the contrary, driver's decisions are governed by the open quantum system model. If we denote the steady-state solution of Equation \eqref{eq:lindblad} as $OQ_{pure}(s_1;\alpha,\lambda)$ for a given car's strategy $s_1\in\{A,N\}$, the final response of the driver is defined as 
$$FR_2(s_1)=OQ_{pure}(s_1;\alpha,\lambda),$$ 
where $\alpha$ and $\lambda$ are driver's model parameters in Equations \eqref{eq:lindblad} and \eqref{eq:luce} respectively. Then the (pure-strategy) equilibrium of this game is defined as follows.
\begin{defn}
A strategy profile $(s_1^*,s_2^*)\in\{A,N\}\times\{C,S\}$ is a \textbf{pure strategy equilibrium} if and only if $s_1^*=BR_1(s_2^*)$ and $s_2^*=OQ_{pure}(s_1^*;\alpha,\lambda)$. 
\label{Defn: Pure Strategy Equilibrium}
\end{defn}

On the contrary, the concept of mixed strategy equilibrium is actually more natural to the open-quantum-system model since the solution tells the probability of taking various actions instead of indicating a particular action. The open quantum system model directly gives a mixed strategy. Let the mixed strategy of the driver is denoted as $\sigma_2 = (p_C,1-p_C)$ where $p_C$ is the probability that the driver chooses to continue. Similarly, let the car's mixed strategy be denoted as $\sigma_1 = (p_A, 1-p_A)$, where $p_A$ is the probability that the car chooses to alert. Then, a mixed strategy profile is denoted as $(\sigma_1, \sigma_2)$. In such a mixed strategy setting, the car's final response is its best mixed-strategy response, i.e. 
$$FR_1(\sigma_2)=BR_1(\sigma_2).$$ 
Similarly, the final response of the driver is obtained from the steady-state solution of Eq. \ref{eq:lindblad}, i.e. 
$$FR_2(\sigma_1)=OQ_{mix}(\sigma_1;\alpha,\lambda).$$ 
Then the mixed-strategy equilibrium of this game is defined as follows.
\begin{defn}
A strategy profile $(\sigma_1^*,\sigma_2^*)$ is an \textbf{mixed-strategy equilibrium} if and only if $\sigma_1^* = BR_1(\sigma_2^*)$ and $\sigma_2^* = OQ_{mix}(\sigma_1^*;\alpha,\lambda)$. 
\label{Defn: Mixed Strategy Equilibrium}
\end{defn}

\begin{figure}[!t]
\centering
\includegraphics[width=0.49\textwidth]{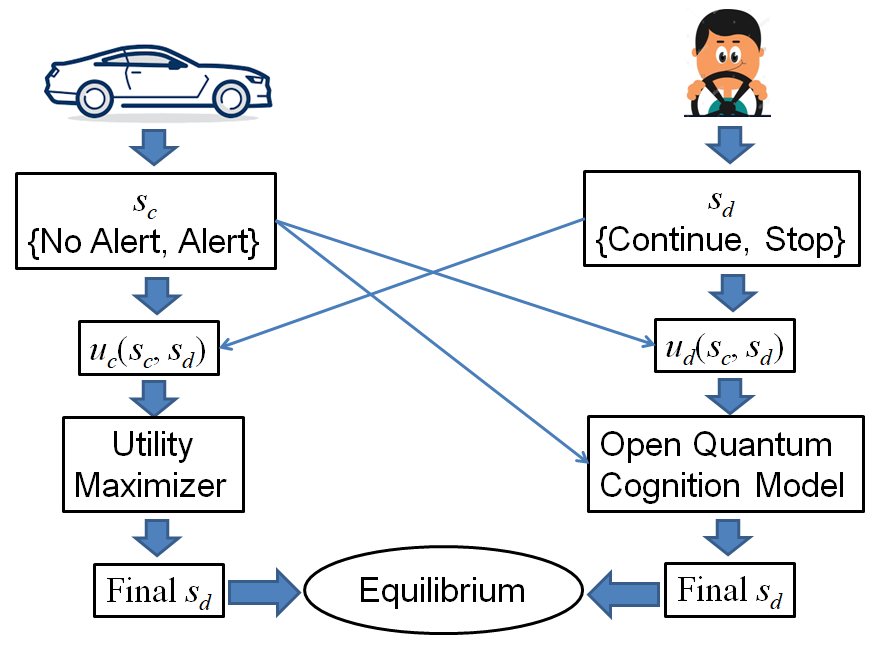}
\caption{Illustration of the car-driver interaction game}
\label{Fig: illustrate}
\end{figure}

Note that the above equilibrium notions presented in Definitions \ref{Defn: Pure Strategy Equilibrium} and \ref{Defn: Mixed Strategy Equilibrium} are novel and different from traditional equilibrium notions in game theory. This is because our game comprises of two different players: (i) the car modeled as an expected utility maximizer, and (ii) the driver modeled using open quantum cognition equation, as is illustrated in Figure \ref{Fig: illustrate}. However, our equilibrium notions are both inspired from the traditional definition of Nash equilibrium, and are defined using players' final responses as opposed to best responses in the Nash sense. By doing so, we can easily expand traditional equilibrium notions to any strategic setting where heterogeneous entities interact in a competitive manner.

\section{Driver's Final Response}
\label{section:lindblad}

Note that the dependent variable $\rho$ in Equation \eqref{eq:lindblad} is a matrix. In order to obtain the analytical solution, we vectorize $\rho$ by stacking its columns one on another to obtain vector $\vv{\rho}$. Thus, the vectorized version of Definition \ref{Defn: Lindblad-Kossakowski} is as follows.
\begin{defn}
The vectorized form for Lindblad-Kossakowski equation is given by
\begin{equation}
\label{eq:vectorize}
\displaystyle \frac{d\vv{\rho}}{dt} = \left[ -i(1-\alpha)\vv{H}+\alpha\vv{L} \right] \vv{\rho},
\end{equation}
where $I_N$ is the $N\times N$ identity matrix,
\begin{equation}\label{eq:vecH}
\vv{H} = H\otimes I_N-I_N\otimes H^T,
\end{equation}
\begin{equation}
\vv{L} = \displaystyle \sum_{m,n} \gamma_{m,n}\Lambda_{m,n}
\end{equation}
\begin{equation}\label{eq:Lambda}
\Lambda_{m,n} = L_{m,n}\otimes L^*_{m,n} - \Phi_{m,n},
\end{equation}
\begin{equation}\label{eq:Phi}
\Phi_{m,n} = \displaystyle \frac{1}{2} \Big( L^\dagger_{m,n}L_{m,n}\otimes I_N +I_N\otimes(L^\dagger_{m,n}L_{m,n})^* \Big),
\end{equation}
with the superscript * representing taking the complex conjugate of all entries.

\label{Defn: Lindblad-Kossakowski vectorized}
\end{defn}

In the driver-car game presented in Section \ref{section:formulation}, note that we have $N=8$ basis states as stated in Equation \eqref{Eqn: Basis States}. We will first derive the sparse structure of $\vv{H}$ in Lemma \ref{Lemma: Lemma 1}. 

Note that the symbol $\oplus$ means direct-sum while the symbol $\otimes$ means tensor-product. The following two simple examples show their difference.
\begin{equation*}
\begin{bmatrix}
a & b \\ c & d
\end{bmatrix}
\oplus
\begin{bmatrix}
e & f \\ g & h
\end{bmatrix}
= 
\begin{bmatrix}
a & b & 0 & 0 \\ c & d & 0 & 0 \\
0 & 0 & e & f \\ 0 & 0 & g & h
\end{bmatrix}
\end{equation*}
\begin{equation*}
\begin{bmatrix}
a & b \\ c & d
\end{bmatrix}
\otimes
\begin{bmatrix}
e & f \\ g & h
\end{bmatrix}
= 
\begin{bmatrix}
ae & af & be & bf \\ ag & ah & bg & bh \\
ce & cf & de & df \\ cg & ch & dg & dh
\end{bmatrix}
\end{equation*}

\begin{lma}
If the Hamiltonian $H$ of the 8-dimensinoal Lindblad-Kossakowski equation is defined as
\begin{equation*}
    H = \begin{bmatrix} 
1 & 1 \\1 & 1\end{bmatrix}\oplus\begin{bmatrix} 
1 & 1 \\1 & 1\end{bmatrix}\oplus\begin{bmatrix} 
1 & 1 \\1 & 1\end{bmatrix}\oplus\begin{bmatrix} 
1 & 1 \\1 & 1\end{bmatrix},
\end{equation*}
then its vectorized form $\vv{H}$ is given by
\begin{equation}
\vv{H}=J\oplus J\oplus J\oplus J,
\end{equation}
where 
\begin{equation*}
J =
\begin{bmatrix} 
X\oplus X\oplus X\oplus X & I_8 
\\
I_8 & X\oplus X\oplus X\oplus X
\end{bmatrix}
\end{equation*}
with 
$X=
\begin{bmatrix} 
0 & -1 
\\
-1 & 0
\end{bmatrix}$. 
\label{Lemma: Lemma 1}
\end{lma}

\begin{proof}
By Equation \eqref{eq:vecH}, we only need to calculate $H\otimes I_N$ and $I_N\otimes H^T$. Noting $N=8$, we have
\begin{equation*}
    H\otimes I_N=\begin{bmatrix} 
I_8 & I_8 \\I_8 & I_8\end{bmatrix}\oplus\begin{bmatrix} 
I_8 & I_8 \\I_8 & I_8\end{bmatrix}\oplus\begin{bmatrix} 
I_8 & I_8 \\I_8 & I_8\end{bmatrix}\oplus\begin{bmatrix} 
I_8 & I_8 \\I_8 & I_8\end{bmatrix}
\end{equation*}
and
\begin{equation*}
    I_N\otimes H^T=I_N\otimes H=K\oplus K\oplus K\oplus K
\end{equation*}
where 
\begin{equation*}
    K=\begin{bmatrix} 
\mathbf{1}\oplus\mathbf{1}\oplus\mathbf{1}\oplus\mathbf{1}& 0 \\0 & \mathbf{1}\oplus\mathbf{1}\oplus\mathbf{1}\oplus\mathbf{1}\end{bmatrix}
\end{equation*}
with \textbf{1} the $2\times 2$ matrix whose elements are all 1.

Subtracting $I_N\otimes H^T$ from $H\otimes I_N$ blockwise then leads to the claimed $J$.
\end{proof}

\begin{remark}
\normalfont
The condition of Lemma 1 is just setting the Hamiltonian of the Lindblad-Kossakowski equation as in Equation \eqref{eq:H}. $\vv{H}$ is a sparse block diagonal matrix with four blocks, each being $J$. $J$ is a sparse matrix consists of four blocks where the off-diagonal blocks are identity matrices and the diagonal matrices are again block diagonal matrices. Such a special structure results from stacking the columns of the all-one matrices.
\end{remark}

Theorem 1 presents the special sparse structure of $\vv{L}$. To prove Theorem 1, Lemma \ref{Lemma: Lemma 2} is needed. Lemma \ref{Lemma: Lemma 2} gives the sparse structure of $\Lambda_{m,n}$.

\begin{lma}\label{Lemma: Lemma 2} 
The $(M,N)^{th}$ entry of the matrix
$\Lambda_{m,n}$ with $m \neq n$ is given by
\begin{equation}
\Lambda_{m,n}(M,N) = 
\begin{cases}
-\displaystyle \frac{1}{2}, & \text{ if } M = N = 8(n-1)+k
\\
 & \text{ or } M = N = 8(k-1)+n, 
 \\
 & k \in \{1, 2, \cdots, 8 \} \setminus \{n\}
\\[2ex]
-1 & \text{ if } M = N = 9n-8
\\[2ex]
1 & \text{ if } M = 9m-8, N = 9n-8
\\[2ex]
0 & \text{ otherwise}
\end{cases}
\end{equation}

The $(M,N)^{th}$ entry of the matrix $\Lambda_{m,n}$ with $m=n$ is given by
\begin{equation}
\Lambda_{m,n}(M,N) = 
\begin{cases}
-\displaystyle \frac{1}{2}, & \text{ if } M = N = 8(n-1)+k
\\
 & \text{ or } M = N = 8(k-1)+n, 
 \\
 & k \in \{1, 2, \cdots, 8 \} \setminus \{n\}
\\[2ex]
0 & \text{ otherwise}
\end{cases}
\end{equation}
\end{lma}

\begin{proof}
$L_{m,n}=\ket{m}\bra{n}$ is a real matrix, so $L^*_{m,n}=L_{m,n}$ and $L^\dagger_{m,n}=L^T_{m,n}$.

Since only the $(m,n)$ entry of $L_{m,n}$ is 1 and the others are 0, $L_{m,n}\otimes L^*_{m,n}$ is a $64\times 64$ matrix with all entries zero except the $(8(m-1)+m,8(n-1)+n)$ entry, which is 1. Note that $m$ and $n$ range from 1 to 8.

Since the $(n, n)$ entry of $L^\dagger_{m,n}L_{m,n}$ is 1 and the other entries are 0, $L^\dagger_{m,n}L_{m,n}\otimes I_8$ is a 64$\times$64 matrix whose entries are all zero except the $[8(n-1)+1]$th to the 8$n$th diagonal entries (which are 1), and $I_8\otimes(L^\dagger_{m,n}L_{m,n})^*$ is a 64$\times$64 matrix whose entries are all zero except the $(M,M)$ entries (which are 1) with $M=8(k-1)+n$, $k=1,2,...,8$ for each fixed $n$. Thus by Equation \eqref{eq:Phi}, $\Phi_{m,n}$ is a 64$\times$64 matrix whose entries are all zero except the $(M,M)$ entries with $M=8(n-1)+k$ or $M=8(k-1)+n$, $k=1,2,...,8$ for each fixed $n$. The $(M,M)$ entries are 1/2 when $8(n-1)+k\neq 8(k-1)+n$ and is 1 when $8(n-1)+k=8(k-1)+n$.

By Equation \eqref{eq:Lambda}, subtracting $\Phi_{m,n}$ from $L_{m,n}\otimes L^*_{m,n}$ leads to the claimed result:  When $m\neq n$, there is no cancellation of nonzero entries between $\Phi_{m,n}$ and $L_{m,n}\otimes L^*_{m,n}$. When $m=n$, only the $(9n-8,9n-8)$ entry of $L_{m,n}\otimes L^*_{m,n}$ is nonzero (which is 1). The $(9n-8,9n-8)$ entry of $\Phi_{m,n}$ also 1. Thus the resultant only has 14 nonzero entries.
\end{proof} 

\begin{remark}
\normalfont
Note that $\Lambda_{m,n}$ does not mean the $(m,n)$ entry of $\Lambda$. $\Lambda_{m,n}$ is itself a matrix. There are 64 such matrices and they will be weighed by $\gamma_{m,n}$ and summed. Then $(M,N)$ entry of $\Lambda_{m,n}$ depend on $m$, $n$, $M$, and $N$. $\Lambda_{m,n}$ is very sparse. The nonzero entries can only take $\pm1$ and $-1/2$ since the building blocks $L_{m,n}$ and $I_N$ only has 1 as nonzero entry value. Given $m$ and $n$, the $(M,N)$ entries with $M=N=8(n-1)+k$ or $M=N=8(n-1)+k$ are special since either $L_{m,n}\otimes L^*_{m,n}$ or $\Phi_{m,n}$ takes nonzero values at these entries.
\end{remark}

Next we will multiply the $\Lambda_{m,n}$ obtained in Lemma \ref{Lemma: Lemma 2} with $\gamma_{m,n}$ and sum over all $m$ and $n$ to obtain $\vv{L}$ in Theorem \ref{Theorem: Theorem 1}.

\begin{thrm}\label{Theorem: Theorem 1} 
Let the coefficients $\gamma_{m,n}$ in the 8-dimensinoal Lindblad-Kossakowski equation be the $(m,n)$ entries of the matrix (ref. to Equation \eqref{eq:PI})
\begin{equation*}
C = \Pi(\lambda) = M[\mu(\lambda)]\oplus M[\nu(\lambda)]\oplus M[\xi(\lambda)]\oplus M[o(\lambda)],
\end{equation*}
where $M[a]$ is of the form
\begin{equation*}
M[a]=
\begin{bmatrix} 
a & a 
\\
1-a & 1-a
\end{bmatrix}.
\end{equation*}
Then, the $(M,N)^{th}$ entries of $\vv{L}$ within the vectorized Lindblad-Kossakowski equation (ref. to Def. \ref{Defn: Lindblad-Kossakowski vectorized}) with $M=N$ are given by
\begin{equation*}
\vv{L}_{M,N} =\begin{cases}
-\displaystyle\frac{1}{2} \ (C_{n+1,n}+C_{l+1,l}), & n\neq l, n \text{ is odd} 
\\[1ex]
-\ C_{n+1,n}, & n=l, n \text{ is odd} 
\\[1ex]
-\displaystyle\frac{1}{2} \ (C_{n-1,n}+C_{l-1,l}), & n\neq l, n \text{ is even} \\[1ex]
-\ C_{n-1,n}, & n=l, n \text{ is even}
\end{cases}
\end{equation*}
where $n = \lfloor \frac{M-1}{8}+1 \rfloor$ and $l= (M-1) \mod 8 + 1$, and the $(M,N)^{th}$ entries of $\Gamma$ with $M \neq N$ are given by
$$
\vv{L}_{M,N} =
\begin{cases}
C_{n+1,n}, & M = 9n+1, N = 9n-8, n =1,3,5,7 
\\[1ex]
C_{n-1,n}, & M = 9n-17, N = 9n-8, n =2,4,6,8
\\[1ex]
0, & \text{otherwise}
\end{cases}
$$
\end{thrm}

\begin{proof}
Interested readers may refer to Appendix \ref{sec: Proof of theorem 1}.
\end{proof}

\begin{remark}
\normalfont $\vv{L}_{M,N}$ depends on $M$ and $N$. The expression of $\vv{L}_{M,N}$ must consist of entries of $C$. Theorem 1 just reveals explicitly these relations. The entries of $C$ appearing in the expression of $\vv{L}_{M,N}$ are $C_{n,n}$ and $C_{n\pm1,n}$ where $n = \lfloor \frac{M-1}{8}+1 \rfloor$ or $n= (M-1) \mod 8 + 1$. Such relations arise due to vectorization (stacking columns). Dividing by 8 and mode 8 appear since each column to be stacked is 8-dimensional. Despite summation over all $m$ and $n$, at most two entries of $C$ appear in $\vv{L}_{M,N}$ since $C=\Pi(\lambda)$ is itself sparse.
\end{remark}

Next we will combine the $\vv{H}$ obtained in Lemma \ref{Lemma: Lemma 1} and the $\vv{L}_{m,n}$ obtained in Theorem \ref{Theorem: Theorem 1} to obtain $-i(1-\alpha)\vv{H}+\alpha\vv{L}$ in Corollary \ref{Corollary: Corollary 1}. 

\begin{cor}\label{Corollary: Corollary 1} 
If the coefficients $\gamma_{m,n}$ of the 8 dimensional Lindblad-Kossakowski equation is set as the $(m,n)$ entries of 
\begin{equation*}
    C=\Pi(\lambda)=M(\mu(\lambda))\oplus M(\nu(\lambda))\oplus M(\xi(\lambda))\oplus M(o(\lambda)),
\end{equation*}
where $M(a)$ is a matrix in the form of
\begin{equation*}
    M(a)=\begin{bmatrix} a & a \\1-a & 1-a\end{bmatrix},
\end{equation*}
then 
\begin{equation*}
    -i(1-\alpha)\vv{H}+\alpha\vv{L} = A_1\oplus A_2\oplus A_3\oplus A_4
\end{equation*}
where 
\begin{equation*}
    A_i=\begin{bmatrix} 
B_{i1}\oplus B_{i2}\oplus B_{i3}\oplus B_{i4} & -i(1-\alpha)I_8+\alpha E_i \\-i(1-\alpha)I_8+\alpha D_i & B_{i5}\oplus B_{i6}\oplus B_{i7}\oplus B_{i8}\end{bmatrix}.
\end{equation*}
$D_i$ and $E_i$ are 16$\times$16 matrices. They both have only one nonzero entry. The nonzero entries are taken from the cognition matrix $C$:
\begin{equation*}
\begin{array}{ccc}
D_1(2,1)=C_{2,1}, & 
E_1(1,2)=C_{1,2}, &
D_2(4,3)=C_{4,3},
\\[2ex]
E_2(3,4)=C_{3,4}, & 
D_3(6,5)=C_{6,5}, &
E_3(5,6)=C_{5,6},
\\[2ex]
D_4(8,7)=C_{8,7}, & 
E_4(7,8)=C_{7,8}. & 
\end{array}
\end{equation*}

The $B_{ij}$'s are 4$\times$4 matrices:

\begin{equation*}
B_{ii} = \displaystyle F_i-\frac{\alpha}{2}\begin{bmatrix} 
C_{2i,2i-1}-C_{2i-1,2i-1}  & 0 \\0 & C_{2i-1,2i}+C_{2i,2i}\end{bmatrix}, 
\end{equation*}
\begin{equation*}
B_{i(i+4)}=G_i-\frac{\alpha}{2}\begin{bmatrix} 
C_{2i-1,2i-1}+C_{2i,2i-1}  & 0 \\0 & C_{2i-1,2i}-C_{2i,2i}\end{bmatrix}, 
\end{equation*}
\begin{equation*}
B_{ij}=F_i-\frac{\alpha}{2}\begin{bmatrix} 
C_{2j-1,2j-1}+C_{2j,2j-1} & 0 \\0 & C_{2j-1,2j}+C_{2j,2j}\end{bmatrix},
\end{equation*}
\begin{equation*}
B_{i(j+4)}=G_i-\frac{\alpha}{2}\begin{bmatrix} 
C_{2j-1,2j-1}+C_{2j,2j-1} & 0 \\0 & C_{2j+1,2j}+C_{2j,2j}\end{bmatrix}
\end{equation*}
for $i\neq j$, $i=1,2,3,4$, $j=1,2,3,4$, where
\begin{equation*}
F_i=i(1-\alpha)\begin{bmatrix} 0 & 1 \\1 & 0\end{bmatrix}-\frac{\alpha}{2}(C_{2i-1,2i-1}+C_{2i,2i-1})I_2,
\end{equation*}
\begin{equation*}
G_i=i(1-\alpha)\begin{bmatrix} 0 & 1 \\1 & 0\end{bmatrix}-\frac{\alpha}{2}(C_{2i-1,2i}+C_{2i,2i})I_2.
\end{equation*}
\end{cor}

\begin{remark}
\normalfont
The Lindblad-Kossakowsi equation itself is not a cognition model since its coefficients $\gamma_{m,n}$ are quite general. The open quantum cognition model is built by setting the $\gamma_{m,n}$ as $(m,n)$ entry of the cognition matrix $C$. The condition in Corollary 1 is just setting $\phi=0$ in Equation \eqref{eq:cognition matrix} and using the $\Pi(\lambda)$ prescribed in Equation \eqref{eq:PI}. This is exactly the scenario of the car-driver game. 
\end{remark}

\begin{remark}
\normalfont
The vectorized operator $-i(1-\alpha)\vv{H}+\alpha\vv{L}$ of the vectorized Lindblad-Kossakowski equation is a block diagonal matrix with four blocks. The four blocks have very similar structures. Each block is actually quite sparse since each block is a block matrix with totally four sub-blocks and the two off-diagonal sub-blocks are almost identity matrix (only one entry is different). 
\end{remark}

\section{Pure-strategy equilibrium}
\label{section:pure}

The diagonal elements of the steady-state solution $\rho$ of Equation \eqref{eq:lindblad} are just $Pr(SNC), Pr(SNS), \cdots, Pr(DAS)$. Then we can calculate the probability for the driver to continue as 
\begin{equation}\label{eq:PrCsum}
Pr(C) = Pr(SNC) + Pr(SAC) + Pr(DNC) + Pr(DAC).
\end{equation}

Let $p$ be the probability that the driver judges the road to be safe before knowing the car's action and $U_2$ be the utility function of the driver. In this paper, we model driver's pure strategy $s_2$ as the output of the open quantum cognition model parameters $\alpha$ and $\lambda$ taking the pure strategy of the car $s_1$ as input:  
\begin{equation}
s_2 = OQ_{pure}(s_1;\Theta) = \begin{cases}
C, & \text{ if } Pr(C) \geq 0.5,
\\[2ex]
S, & \text{ if } Pr(C) < 0.5.
\end{cases}
\end{equation}
where $\Theta=(\alpha,\lambda,p,U_2)$ is the parameter tuple of the open quantum model. 

\begin{remark}
\normalfont
In this paper, we use $Pr(C)$ in two different ways to obtain pure and mixed strategy equilibria. We obtain a pure strategy at the driver by employing a hard threshold on $Pr(C)$ (in our case, Continue if $Pr(C) \geq 0.5$, Stop otherwise). By treating $Pr(C)$ as the driver's mixed strategy in Section \ref{section:mix}, we will obtain the mixed-strategy equilibrium. 
\end{remark}

We set the initial density matrix as $\rho_0=\ket{\Psi_0}\bra{\Psi_0}$, where 
\begin{equation*}
\ket{\Psi_0}=\sqrt{p/2}(\ket{e_3}+\ket{e_4})+\sqrt{(1-p)/2}(\ket{e_7}+\ket{e_8})
\end{equation*}
when the car action is A and  
\begin{equation*}
\ket{\Psi_0}=\sqrt{p/2}(\ket{e_1}+\ket{e_2})+\sqrt{(1-p)/2}(\ket{e_5}+\ket{e_6})
\end{equation*}
when the car action is N, with $\ket{e_i}$ prescribed in Subsection \ref{subsec: oq}. The calculation of the generalized pure-strategy equilibrium is similar to that of the Nash equilibrium. We simply replace the best response with the final response. We loop over the car strategies. In the loop, the car strategy is the input of the open quantum model and a driver strategy is the output. If the car strategy is the best response with the outputted driver strategy, then the strategy profile is outputted as pure-strategy equilibrium. Algorithm \ref{algo_pure} lists the procedures of calculating the pure-strategy equilibrium.

\begin{figure}[!t]
\centering
\begin{subfigure}{0.24\textwidth}
\centering
\includegraphics[width=\textwidth]{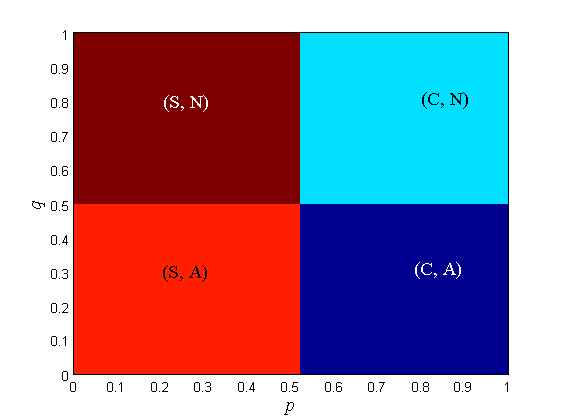}
\caption{Driver-Agnostic car}
\end{subfigure}%
\begin{subfigure}{0.24\textwidth}
\centering
\includegraphics[width=\textwidth]{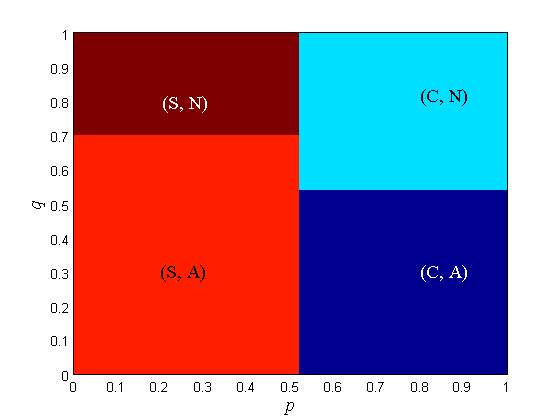}
\caption{Driver-Conscient car}
\end{subfigure}%
\caption{Equilibrium points of the driver-car games with a driver-agnostic car and with a driver-conscient car (i.e., assumes that the driver uses open quantum model with $\lambda= 10, \alpha=0.2$ to make decisions) under various prior beliefs.}
\label{fig:pureAlpha}
\end{figure}

\begin{table}[!t]
\centering
\begin{tabular}{c}
\begin{tabular}{r|c|c|}
\multicolumn{1}{r}{} &  \multicolumn{1}{c}{C} & \multicolumn{1}{c}{S} 
\\
\cline{2-3}
N & $a_{1,s} = 85$, $a_{2,s} = 85$ & $b_{1,s} = 75$, $b_{2,s} = 50$ 
\\
\cline{2-3}
A & $c_{1,s} = 40$, $c_{2,s} = 85$ & $d_{1,s} = 50$, $d_{2,s} = 50$ 
\\
\cline{2-3}
\end{tabular}
\\[5ex]
\begin{tabular}{ r|c|c| }
\multicolumn{1}{r}{}
 &  \multicolumn{1}{c}{C}
 & \multicolumn{1}{c}{S} \\
\cline{2-3}
N & $a_{1,d} = 25$, $a_{2,d} = 25$ & $b_{1,d} = 30$, $b_{2,d} = 60$ \\
\cline{2-3}
A & $c_{1,d} = 75$, $c_{2,d} = 25$ & $d_{1,d} = 85$, $d_{2,d} = 85$ \\
\cline{2-3}
\end{tabular}
\end{tabular}
\caption{Utilities used in our numerical results when the road is safe (above) and dangerous (below)}
\label{tab:tab1-sim}
\end{table}

Furthermore, in our numerical evaluation, we assume the utilities at both the car and the driver as shown in Table \ref{tab:tab1-sim}. In addition to the case of a driver-conscient car, we consider a benchmark case where the car does not care about the driver and makes decisions solely based on its prior, i.e., alert if $q<0.5$ and does not alert if $q\geq0.5$. In this benchmark case, the final response of the car is independent of the driver's strategy. The equilibrium points of the driver-car games with a driver-agnostic car and with a driver-conscient car (driver making decisions according to open quantum model with $\lambda= 10, \alpha=0.2$) under various prior beliefs on road condition are shown in Fig. \ref{fig:pureAlpha}. When both the driver and car are sure of safety, the equilibrium is $(N, C)$. When both the driver and car are sure of danger, the equilibrium is $(A, S)$. When the driver is sure of safety but the car is sure of danger, the equilibrium is $(A, C)$. When the driver is sure of danger but the car is sure of safety, the equilibrium is $(N, S)$. The division line is not $p = q = 0.5$. (S, A) has the largest area. When the car is driver-agnostic, the border between Not Alert and Alert in the equilibrium plot is always $q=0.5$ regardless of the equilibrium strategy of the driver. When the car is driver-conscient, the border between Not Alert and Alert depends on the equilibrium strategy of the driver (or equivalently, road prior of the driver): the border is located close to $q=0.7$ when $p\leq0.50$ and the border is located close to $q=0.52$ when $p\geq0.52$.

\IncMargin{1em}
\begin{algorithm}[!t]
\SetKwData{Left}{left}
\SetKwData{This}{this}
\SetKwData{Up}{up}
\SetKwFunction{Union}{Union}\SetKwFunction{FindCompress}{FindCompress}
\SetKwInOut{Input}{input}\SetKwInOut{Output}{output}

\Input{Parameters: $\alpha,\lambda$; prior about road safety of the driver: $p$; prior about road safety of the car: $q$; Utility function of the car: $U_1(s_1,s_2,r)$ where $r=S$ means safe and $r=D$ means dangerous; Utilities of the driver: $U_2$.}
\Output{Pure-strategy equilibrium: $S^*$}
\BlankLine
\tcp{Empty set means no equilibrium}
$S^*=\emptyset$\;
\For{$s_1\;\mathbf{in}\;\{N,A\}$}{
 $\Bar{s}_1 = \mathrm{element\;of\;}\{N,A\}\setminus\{s_1\}$\;
 $s_2 = OQ_{pure}(s_1;\alpha,\lambda)$\;
 $u = qU_1(s_1,s_2,S)+(1-q)U_1(s_1,s_2,D)$\;
 $\Bar{u} = qU_1(\Bar{s}_1,s_2,S)+(1-q)U_1(\Bar{s}_1,s_2,D)$\;
 \If{$u\geq\Bar{u}$}{
    $S^*=S^*\cup\{s_1,s_2\}$\;
  }
 }

\caption{Calculating the pure-strategy equilibrium of the car-driver game}
\label{algo_pure}
\end{algorithm}
\DecMargin{1em}

The equilibrium points of the driver-car game with $\lambda$ = 0, 1, 2, 3, 4, 10 and $\alpha$ = 0.8 under various prior beliefs on road condition are shown in Fig. \ref{fig:pureLambda} (C: Continues, S: Stop, A: Alert, N: Not Alert). When $\lambda$ drops from 10 to 4, the border between S and C shifts from left to right. When $\lambda$ drops from 4 to 3, the border between (S, A) and (C, A) shifts from left to right and a region with two equilibrium points appears inside the region of (C, N). The two equilibria are (C, N) and (S, A). When $\lambda$ drops from 4 to 3 and from 3 to 2, the border between (S, A) and (C, A) shifts from left to right and the region with two equilibrium points enlarges with the border shift. When $\lambda$ drops from 2 to 0, the driver can no longer distinguish the utilities. The $(C, N)$ region is merged into the $(S, N)$ region and the two-equilibrium region is merged into the $(S, A)$ region. The border between $(S, A)$ and $(C, A)$ shifts from left to right and a new no-equilibrium region appears inside the previous $(C, N)$ region.

\begin{figure*}[htp]
    \centering
    \begin{subfigure}{0.25\textwidth}
      \centering
      \includegraphics[width=\textwidth]{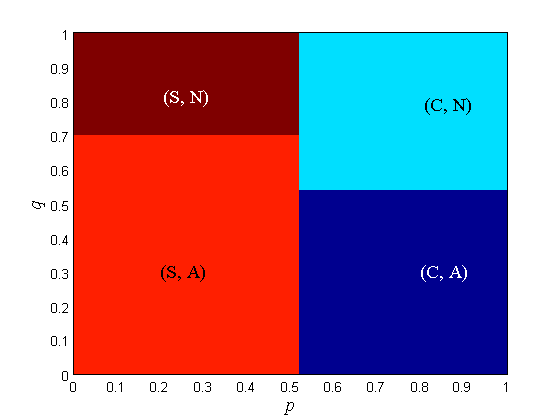}
      \caption{$\lambda=10$}
    \end{subfigure}%
    \begin{subfigure}{0.25\textwidth}
      \centering
      \includegraphics[width=\textwidth]{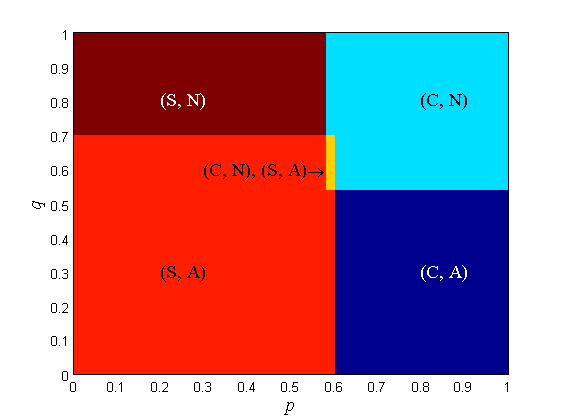}
      \caption{$\lambda=3$}
    \end{subfigure}%
    \begin{subfigure}{0.25\textwidth}
      \centering
      \includegraphics[width=\textwidth]{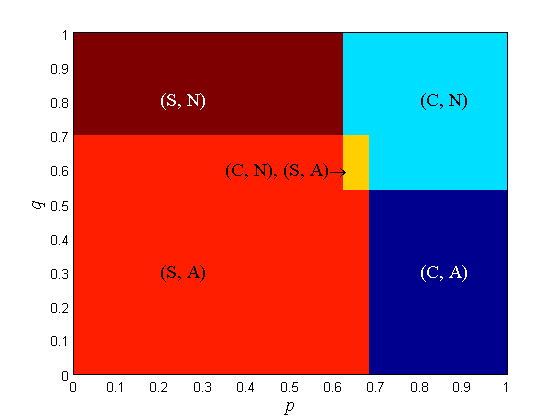}
      \caption{$\lambda=1$}
    \end{subfigure}%
  \caption {Equilibrium points of the driver-car game with $\alpha=0.8$ under various prior beliefs on road condition}
  \label{fig:pureLambda}
\end{figure*}

\begin{remark}
\normalfont
When $\lambda=0$, the driver cannot distinguish the utilities at all and is completely random, so the concept of final response does not apply. The type of pure-strategy equilibrium strongly aligns with the priors of the driver and the car. The desired equilibria are $(C, N)$ and $(S, A)$, where the driver's action is in harmony with car's action. 
\end{remark}

\begin{remark}
\normalfont
Since Fig. \ref{fig:pureAlpha} and Fig. \ref{fig:pureLambda} are plotted over $(p,q)$ axes, we can find out which type of equilibrium is most common. With the prescribed utilities, the most common pure-strategy equilibrium is $(S, A)$. This is the most favorable equilibrium, since following the car's recommendation in the dangerous road can save life. 
\end{remark}

\begin{remark}
\normalfont
As the driver's ability to distinguish utilities weakens ($\lambda$ decreases), $(S, A)$ becomes more likely. This means that the driver follows the car's advice diligently especially when he/she is incapable of making decisions on a dangerous road. 
\end{remark}

\section{Mixed-strategy equilibrium}
\label{section:mix}

When calculating the mixed-strategy equilibrium, $p$ and $p_A$ appear in the initial state of the open quantum model since the mixed-strategy of the car is completely determined by $p_A$ (ref. to Subsecion \ref{section:eq}). Theorem 2 will give a closed-form expression of $Pr(C)$ by solving the vectorized Lindblad-Kossakowski equation (ref. to Definition \ref{Defn: Lindblad-Kossakowski vectorized}). 

\begin{figure*}
\begin{thrm}
Let the initial density matrix be given as $\rho_0=\ket{\Psi_0}\bra{\Psi_0},$ where 
\begin{equation*}
\ket{\Psi_0}=\sqrt{p(1-p_A)/2}(\ket{e_1}+\ket{e_2})+\sqrt{pp_A/2}(\ket{e_3}+\ket{e_4})
+\sqrt{(1-p)(1-p_A)/2}(\ket{e_5}+\ket{e_6})
+\sqrt{(1-p)p_A/2}(\ket{e_7}+\ket{e_8}).
\end{equation*}
The probability that the driver chooses to continue is
\begin{equation*}
\displaystyle Pr(C) = \frac{2(1-\alpha)^2}{c} + \frac{\alpha^2}{c}r + rh(t)
+\Big(\frac{1}{2}-\frac{2(1-\alpha)^2}{c}\Big)e^{-\alpha t} \cos \left[ 2(1-\alpha)t \right]
-\frac{\alpha(1-\alpha)}{c}e^{-\alpha t} \sin \left[ 2(1-\alpha)t \right],
\end{equation*}
where $c=\alpha^2+4(1-\alpha)^2$,
$r = 
p (1-p_A) C_{1,2}+ p p_A C_{3,4} 
+ (1-p)(1-p_A)C_{5,6} + (1-p)p_A C_{7,8}$,
and
\begin{equation*}
h(t) = \displaystyle \frac{\alpha}{c}e^{-\alpha t} \Big\{ 2(1-\alpha)\sin \left[ 2(1-\alpha)t \right]
- \alpha \cos \left[2(1-\alpha)t \right] \Big\}.
\end{equation*}
\vspace{-2ex}
\label{Thrm: Pr(C)}
\end{thrm}
\begin{proof}
Interested readers may refer to Appendix \ref{sec: Proof of Pr(C)}.
\end{proof}
\vspace{-2ex}
\end{figure*}

\begin{remark}
\normalfont
The output of the open quantum model, $Pr(C)$, presented in Theorem \ref{Thrm: Pr(C)} consists of both transient and stationary parts. The transient part consists of sine and cosine multiplied with exponential decay. Thus, there always exists steady state when the exponential decay rate $\alpha > 0$. When $\alpha = 0$ or $\lambda = 0$, $Pr(C)=0$ and the driver is completely random and absent-minded. Furthermore, a driver with a higher $\alpha$ can make a decision faster. Thus $\alpha$ also represents the brain power and attentiveness of the driver. On the other hand, the parameter $\lambda$ appears only within the $C_{ij}$ terms, which are linearly weighted by monomial terms of $p$ and $p_A$. This is essentially prior probability multiplied with likelihood, or initial probability multiplied with transition probability. Viewing $\alpha$ and $p_A$ as constants, the steady-state $Pr(C)$ is a linear function of the driver's prior $p$.
\end{remark}

For brevity, the mixed-strategy equilibrium is denoted as ($p_A^*, p_C^*$). If the car knows that the driver will play $p_C^*$, then its expected payoffs from playing alert and no alert must be equal, otherwise, it either chooses to alert only or chooses not to alert only and does not need to mix between them. Thus we have
\begin{equation}
\begin{split}
p_C^* [a_{1, s}q + a_{1, d}(1-q)] + (1- p_C^*)[b_{1, s}q + b_{1, d}(1-q)]
    \\= p_C^* [c_{1, s}q + c_{1, d}(1-q)] + (1- p_C^*)[d_{1, s}q + d_{1, d}(1-q)].
\end{split}
\end{equation}

Solving for $p_C^*$, we obtain
\begin{equation}
p_C^* = \frac{q\Delta_s + (1-q)\Delta_d}{q (\Delta_s +c_{1, s} -a_{1, s})+ (1-q) (\Delta_d +c_{1, d} -a_{1, d} )}.
\label{Eqn: p_C^*}
\end{equation}
where $\Delta_s = b_{1, s}-d_{1, s}$ and $\Delta_d = b_{1, d}-d_{1, d}$.

For the sake of illustration, we consider the bi-matrix game presented in Table \ref{tab:tab1}. Upon substituting the utility values in Table \ref{tab:tab1-sim} for this example in Equation \eqref{Eqn: p_C^*}, we obtain
\begin{equation}
\begin{array}{lcl}
p_C^* & = & \displaystyle \frac{11-16q}{3q+1}.
\end{array}
\end{equation}
Note that the above $p_C^*$ maybe outside [0, 1]. If so, there is no mixed-strategy equilibrium. In order for $p_C^*$ to lie within [0, 1] under the prescribed utilities, $q$ must lie within [10/19, 11/16]. This is a very narrow range of $q$. Given $p_C^*$, the car can assign any $p_A$ to A because A and N give the same payoff. Next we need to search the $p_A$ that produce $p_C^*$. Such a $p_A$ is just the desired $p_A^*$.

Since $p_C^*$ is completely determined by $q$, it is more convenient to plot $p_A^*$ versus $p$ and $p_C^*$. $p_A^*$ versus $p$ and $p_C^*$ with various $\lambda$s is shown in Fig. \ref{fig:mix}. The mixed-strategy equilibria only exist in a narrow band extending from a low-$p_C^*$-low-$p$ region to a high-$p_C^*$-high-$p$ region. There may not exist a mixed -strategy equilibrium for a given $q$, but there always exists one for a given $p$. When $p_C^*$ and $p$ increase, the band gets narrower. Within the band, the gradient of $p_A^*$ is perpendicular to the band, i.e., $p_A^*$ increases when $p$ increases and $p_C^*$ decreases simultaneously. When $\lambda$ decreases, the band gets flatter and the band firstly widens and then narrows. As $\alpha$ decreases, the band gets flatter and narrower.

\begin{figure*}[htp]
    \centering
    \begin{subfigure}{0.25\textwidth}
      \centering
      \includegraphics[width=\textwidth]{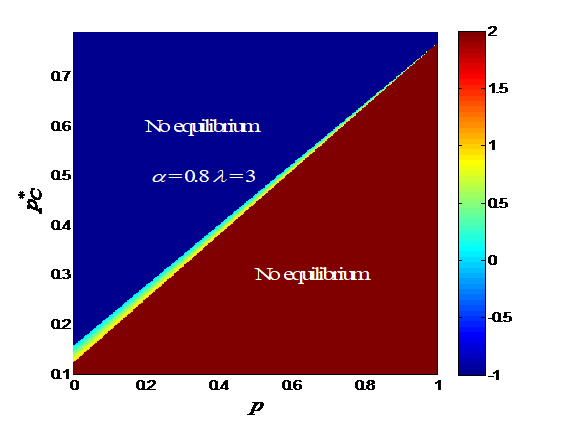}
    \end{subfigure}%
    \begin{subfigure}{0.25\textwidth}
      \centering
      \includegraphics[width=\textwidth]{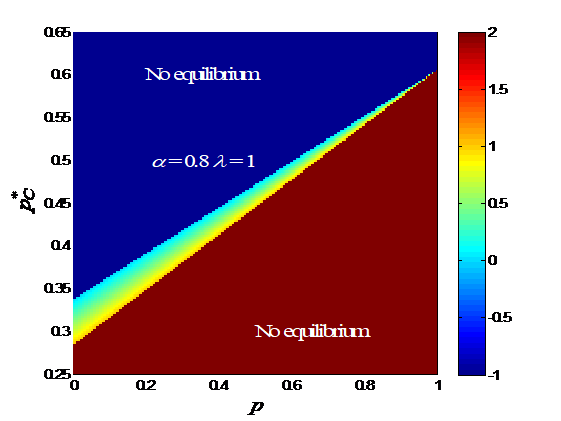}
    \end{subfigure}%
    \begin{subfigure}{0.25\textwidth}
      \centering
      \includegraphics[width=\textwidth]{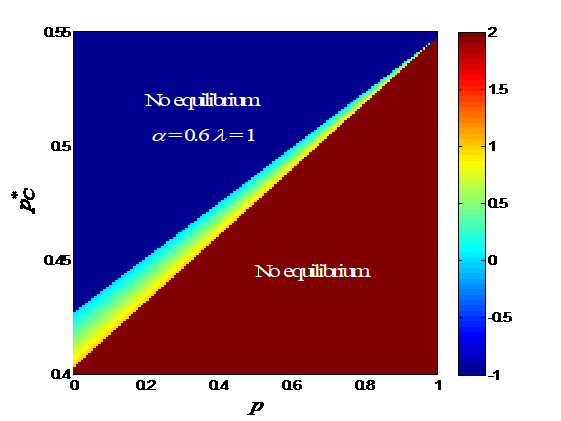}
    \end{subfigure}%
    \begin{subfigure}{0.25\textwidth}
      \centering
      \includegraphics[width=\textwidth]{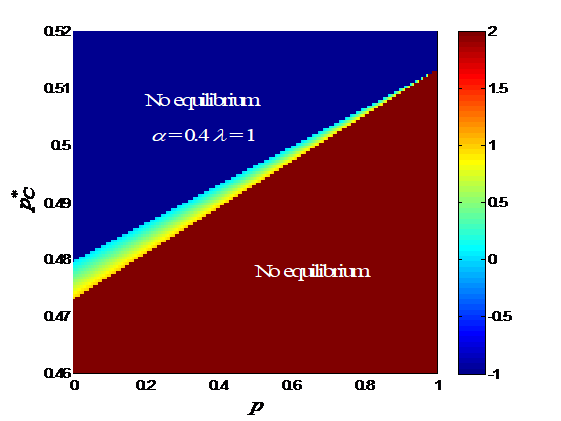}
    \end{subfigure}
  \caption {Existence of $p_A^*$ on the plane of $p$ and $p_C^*$ for different values of $\lambda$ and $\alpha$}
  \label{fig:mix}
\end{figure*}

\begin{remark}
\normalfont 
In the case of mixed equilibria, when the driver is attentive, the equilibrium strategy $P^*_C$ is well aligned with her prior (higher $p$, higher $P^*_C$) as shown in Figure \ref{fig:mix}. However, when the driver gradually loses her attention ($\alpha$ or $\lambda$ decreases), $P^*_C$ steadily approaches to 0.5 regardless of $p$. This means that the driver becomes uncertain to choose $C$ or $S$ at equilibrium, when she is inattentive.
\end{remark}

\section{Conclusion and Future Work}
\label{section:conclusion}
In this paper, we developed a strategic driver-assist system based on a novel vehicle-driver interaction game. While the car is modeled as an expected utility maximizer, the driver is characterized by open-quantum cognition model which models his/her attentiveness ($\alpha$) as well as sensitivity to the utilities ($\lambda$). Based on a novel equilibrium concept proposed to solve any general human-system interaction game, we showed that both the car and the driver employ a threshold-based rule on their respective priors regarding the road state at equilibrium. Through numerical results, we also demonstrated how these thresholds vary under different settings based on road conditions and agent behavior. Specifically, in our proposed framework, we showed that an inattentive driver $(\lambda \leq 1)$ would stop the car in about 65\% of all possible belief profile settings, and at least in 77\% of belief profiles settings when the car alerts the driver. On the contrary, if there were no driver-assist system in the car, an inattentive driver would have stopped the car only in 50\% of all possible belief profiles settings (the region where $p < 0.5$), and in about 38.5\% of all scenarios if the car were to alert the driver using our driver-assist system. 
At the same time, our proposed driver-assist system has improved persuasive ability by taking into account driver behavior, in addition to its inferences regarding the road state. This improvement in performance was demonstrated by the increase in threshold on a driver-conscient car's belief, as opposed to that of a driver-agnostic car.
Furthermore, we also proved that there always exists a mixed strategy equilibrium for any given driver's prior, but only under a small range within car's prior values. As the driver loses attention, we demonstrated that the mixed strategy of the driver at equilibrium drifts towards uniformly distributed probabilistic decisions. In the future, we will investigate repeated interaction games where the car can learn driver's model parameters over multiple iterations. Furthermore, we will also incorporate $B$ matrix 
in the Lindblad model to account for the effects of mental deliberation in resolving conflicts between his/her own prior and the car's signal.

\section{Acknowledgment}

Dr. S. N. Balakrishnan just passed away before this paper is ready to submit. However, Dr. Balakrishnan participated in every aspect through the research related to this paper. Therefore, we decided to still keep Dr. Balakrishnan as the co-author of this paper.

\bibliographystyle{unsrt}
\bibliography{game}
\newpage
\appendices

\section{Proof of Theorem \ref{Theorem: Theorem 1}\label{sec: Proof of theorem 1}}

There are totally 64 terms of $\gamma_{m,n}\Lambda_{m,n}$ to be summed in Equation \eqref{eq:vectorize} when $\gamma_{m,n}=C_{m,n}$. We only need to consider terms with $C_{m,n}\neq 0$. The nonzero entries of $C_{m,n}$ are: $(n,n)$ and $(n+1,n)$ when $n$ is odd, $(n,n)$ and $(n-1,n)$ when $n$ is even. There are only 16 nonzero terms. 

Consider the $(n,n)$ terms. According to Lemma \ref{Lemma: Lemma 2}, $C_{n,n}\Lambda_{n,n}$ has only 14 nonzero entries. They are all $-C_{n,n}/2$ at the $(M,M)$ entries where $M=8(n-1)+k_1$ or $M=8(k_2-1)+n$, $k_1$, $k_2=1,2,...,8$ but $k_1$ and $k_2$ cannot be $n$ at the same time. Each $(M,M)$ entry of $\sum_{n=1}^8C_{n,n}\Lambda_{n,n}$ has exactly two $C_{n,n}\Lambda_{n,n}$'s contributing to it: one with $n = (M-1) \mod 8 + 1$, and the other with $n = \lfloor \frac{M-1}{8}+1 \rfloor$. That is, the $(M,M)$ entry of $\sum_{n=1}^8C_{n,n}\Lambda_{n,n}$ is
\begin{equation*}
\Big(\sum_{n=1}^8C_{n,n}\Lambda_{n,n}\Big)_{M,M} =
\begin{cases}
      (C_{n,n}+C_{l,l})/2, & n\neq l \\
      0, & n=l
\end{cases}
\end{equation*}
where $n = \lfloor \frac{M-1}{8}+1 \rfloor$ and $l = (M-1) \mod 8 + 1$. Note that $n=l$ when $M=9n-8$.

Next consider the $(n\pm 1,n)$ terms. According to Lemma \ref{Lemma: Lemma 2}, $C_{n\pm 1,n}\Lambda_{n\pm 1,n}$ has 16 nonzero entries. 14 of them are $(M,M)$ entries taking value $-C_{n\pm 1,n}/2$ where $M=8(n-1)+k_1$ or $M=8(k_2-1)+n$, $k_1$, $k_2=1,2,...,8$ but $k_1$ and $k_2$ cannot be $n$ at the same time. One of them is $-C_{n\pm 1,n}$ at the $(9n-8,9n-8)$ entry. The last one is $C_{n\pm 1,n}$ at the $(9(n\pm 1)-8, 8(n-1)+n)$ entry. We only need to sum the $C_{n+1,n}\Lambda_{n+1,n}$ terms when $n$ is odd and only need to sum the $C_{n-1,n}\Lambda_{n-1,n}$ terms when $n$ is even. The $(M,M)$ entry of the sum of all $C_{n\pm 1,n}\Lambda_{n\pm 1,n}$ terms is
\begin{equation*}
\begin{cases}
      -(C_{n+1,n}+C_{l+1,l})/2, & n\neq l, n \text{ is odd} \\
      -C_{n+1,n}, & n=l, n \text{ is odd} \\
      -(C_{n-1,n}+C_{l-1,l})/2, & n\neq l, n \text{ is even} \\
      -C_{n-1,n}, & n=l, n \text{ is even}
\end{cases}
\end{equation*}
where $n = \lfloor \frac{M-1}{8}+1 \rfloor$ and $l = (M-1) \mod 8 + 1$.

Summing the $C_{n,n}\Lambda_{n,n}$ terms and the $C_{n\pm 1,n}\Lambda_{n\pm 1,n}$ terms of $\vv{L}$ leads to the claimed result.

\section{Proof of Theorem \ref{Thrm: Pr(C)}\label{sec: Proof of Pr(C)}}

Let $\beta=-i(1-\alpha)$. The solution is in the form of $\vv{\rho}_t=\vv{\rho}(t)=\mathrm{exp}((\beta\vv{H}+\alpha\vv{L})t)\rho_0$. 

By Corollary 1,
\begin{equation}
\begin{array}{l}
\exp((\beta\vv{H}+\alpha\vv{L})t) = \displaystyle \sum^\infty_{k=0}(\beta\vv{H}+\alpha\vv{L})^kt^k/k!
\\[2ex]
\qquad = \ \displaystyle \sum^\infty_{k=0}A_1^kt^k\oplus A_2^kt^k\oplus A_3^kt^k\oplus A_4^kt^k/k!
\\[3ex]
\qquad = \ \displaystyle \mathrm{exp}(A_1t)\oplus\mathrm{exp}(A_2t)\oplus\mathrm{exp}(A_3t)\oplus\mathrm{exp}(A_4t).
\end{array}
\end{equation}

By Equation \eqref{eq:PrCsum} and with the indices defined in Equation \eqref{eq:PI}, 
\begin{equation}
Pr(C) = \displaystyle \vv{\rho}_t(1)+\vv{\rho}_t(19)+\vv{\rho}_t(37)+\vv{\rho}_t(55) 
\end{equation}
where $\vv{\rho}_t(i)$ is the \textit{i}th element of vector $\vv{\rho}_t$.

Calculating $\vv{\rho}_t(1)$  only needs the first row of $\mathrm{exp}(A_1t)$.

We will calculate $\mathrm{exp}(A_1t)$ from $(A_1t)^k$. Again only the first row of $(A_1t)^k$ is needed. Given $(A_1t)^k$, to calculate the first row of $(A_1t)^{(k+1)}$  = $(A_1t)^kA_1t$, still only the first row of $(A_1t)^k$ is needed. Let the first row of $(A_1t)^k$ be denoted as $v_k$. Then $v_{k+1} = v_k A_1$. Denote the \textit{j}th element of $v_k$ as $v_{k,j}$. Then by Corollary \ref{Corollary: Corollary 1}, 
\begin{equation}
\label{eq:v1}
\begin{array}{l}
\;\;\;\;[v_{k+1,2j-1}\; v_{k+1,2j}]
\\[1ex]
=[v_{k,2j-1}\; v_{k,2j}]B_{1j} + [v_{k,2j+7}\; v_{k,2j+8}]F_j,
\end{array}
\end{equation}
\begin{equation}
\label{eq:v2}
\begin{array}{l}
\;\;\;\;[v_{k+1,2j+7}\; v_{k+1,2j+8}]
\\[1ex]
=[v_{k,2j-1}\; v_{k,2j}]G_j + [v_{k,2j+7}\; v_{k,2j+8}]B_{1(j+4)},
\end{array}
\end{equation}
where $j=1,2,3,4$, 
\begin{equation*}
F_1=\begin{bmatrix} \beta & 0 \\\alpha C_{2,1} & \beta\end{bmatrix},\;G_1=\begin{bmatrix} \beta & \alpha C_{1,2} \\0 & \beta\end{bmatrix} 
\end{equation*}
and $F_j=G_j=I_2$ for $j=2,3,4$.

Combining Equation \eqref{eq:v1} and Equation \eqref{eq:v2}, we have
\begin{equation}
\label{eq:iter}
\begin{array}{l}
\;\;\;\;[v_{k+1,2j-1}\; v_{k+1,2j}\; v_{k+1,2j+7}\; v_{k+1,2j+8}]
\\[1ex]
=[v_{k,2j-1}\; v_{k,2j}\; v_{k,2j+7}\; v_{k,2j+8}]\begin{bmatrix} B_{1j} & G_j \\F_j & B_{1(j+4)}\end{bmatrix}.
\end{array}
\end{equation}

Use Equation \eqref{eq:iter} from 1 to $n$, we have
\begin{equation}
\label{eq:vall}
\begin{array}{l}
\;\;\;\;[v_{k,2j-1}\; v_{k,2j}\; v_{k,2j+7}\; v_{k,2j+8}]
\\[1ex]
=[v_{1,2j-1}\; v_{1,2j}\; v_{1,2j+7}\; v_{1,2j+8}]\begin{bmatrix} B_{1j} & G_j \\F_j & B_{1(j+4)}\end{bmatrix}^{k-1}.
\end{array}
\end{equation}

By Corollary \ref{Corollary: Corollary 1}, 
\begin{equation*}
v_1= \begin{bmatrix}
-\alpha C_{2,1} & -\beta & 0_{1\times6} & \beta & \alpha C_{1,2}
\end{bmatrix}
\end{equation*}
where $0_{1\times k}$ is a $1\times k$ column vector. Plugging the initial values in Equation \eqref{eq:vall}, we have
\begin{equation}\label{eq: B11}
\begin{array}{l}
\;\;\;\;[v_{k,1}\; v_{k,2}\; v_{k,9}\; v_{k,10}]
\\[1ex]
=[-\alpha C_{2,1}\; -\beta\; \beta\; \alpha C_{1,2}]\begin{bmatrix} B_{11} & G_1 \\F_1 & B_{15}\end{bmatrix}^{k-1}.
\end{array}
\end{equation}
The other $v_{k,j}$'s with $j\neq 1,2,9,10$ are all zero.

Similarly, calculating $\vv{\rho}_t(19)$, $\vv{\rho}_t(37)$, and $\vv{\rho}_t(55)$ only needs the 3rd row of $\mathrm{exp}(A_2t)$, the 5th row of $\mathrm{exp}(A_3t)$, the 7th row of $\mathrm{exp}(A_4t)$, respectively. Denote the 3rd row of $A_2^k$, the 5th row of $A_3^k$, and the 7th row of $A_4^k$ as $x_k$, $y_k$, and $z_k$, respectively. Then $x_{k+1}=x_kA_2$, $y_{k+1}=y_kA_3$, and $z_{k+1}=z_kA_4$. By Corollary \ref{Corollary: Corollary 1}, 
\begin{equation*}
x_1 = \begin{bmatrix}
0 & 0 & -\alpha C_{4,3} & -\beta & 0_{1\times6}
& \beta & \alpha C_{3,4} & 0_{1\times4}
\end{bmatrix},
\end{equation*}
\begin{equation*}
y_1= \begin{bmatrix}
0_{1\times4} & -\alpha C_{6,5} & -\beta & 0_{1\times6} & \beta & \alpha C_{5,6} & 0 & 0
\end{bmatrix},
\end{equation*}
\begin{equation*}
z_1= \begin{bmatrix}
0_{1\times6} & -\alpha C_{8,7} & -\beta & 0_{1\times6} & \beta & \alpha C_{7,8}
\end{bmatrix}.
\end{equation*}

Denote the \textit{j}th elements of $x_k$, $y_k$, and $z_k$ as $x_{k,j}$, $y_{k,j}$, and $z_{k,j}$, respectively. Then the nonzero elements can be calculated as
\begin{equation}
\begin{array}{l}
\;\;\;\;[x_{k,3}\; x_{k,4}\; x_{k,11}\; x_{k,12}]
\\[1ex]
=[-\alpha C_{4,3}\; -\beta\; \beta\; \alpha C_{3,4}]\begin{bmatrix} B_{22} & G_2 \\F_2 & B_{26}\end{bmatrix}^{k-1},
\end{array}
\end{equation}
\begin{equation}
\begin{array}{l}
\;\;\;\;[y_{k,5}\; y_{k,6}\; y_{k,13}\; y_{k,14}]
\\[1ex]
=[-\alpha C_{6,5}\; -\beta\; \beta\; \alpha C_{5,6}]\begin{bmatrix} B_{33} & G_3 \\F_3 & B_{37}\end{bmatrix}^{k-1},
\end{array}
\end{equation}
\begin{equation}
\begin{array}{l}
\;\;\;\;[z_{k,7}\; z_{k,8}\; z_{k,15}\; z_{k,16}]
\\[1ex]
=[-\alpha C_{8,7}\; -\beta\; \beta\; \alpha C_{7,8}]\begin{bmatrix} B_{44} & G_4 \\F_4 & B_{48}\end{bmatrix}^{k-1},
\end{array}
\end{equation}
In Equation \eqref{eq: B11},
\begin{equation}
    \begin{bmatrix} B_{11} & G_1 \\F_1 & B_{15}\end{bmatrix}=\begin{bmatrix} -\alpha C_{2,1} & -\beta & \beta & \alpha C_{1,2} \\-\beta & b_1 & 0 & \beta \\\beta & 0 & b_5 & -\beta\\\alpha C_{2,1} & \beta & -\beta & -\alpha C_{1,2}\end{bmatrix}
\end{equation}
where 
\begin{equation*}
b_1=-\frac{\alpha}{2}(C_{1,1}+C_{2,2}+C_{2,1}+C_{1,2})=-\alpha
\end{equation*}
and 
\begin{equation*}
b_5=-\frac{\alpha}{2}(C_{2,2}+C_{1,1}+C_{1,2}+C_{2,1})=-\alpha.
\end{equation*}
Thus it can be diagonalized as 
\begin{equation}
    \begin{bmatrix} B_{11} & G_1 \\F_1 & B_{15}\end{bmatrix}=P\begin{bmatrix} -\alpha -2\beta &  &  &  \\ & -\alpha +2\beta &  &  \\ &  & -\alpha & \\  &  &  & 0\end{bmatrix}P^{-1}
\end{equation}
where  
\begin{equation*}
P=\begin{bmatrix} 1 & 1 & 0 & \alpha^2C_{1,2}-2\beta^2 \\1 & -1 & 1 & -\alpha\beta (C_{1,2}-C_{2,1}) \\-1 & 1 & 1 & \alpha\beta (C_{1,2}-C_{2,1})\\ -1 & -1 & 0 & \alpha^2C_{2,1}-2\beta^2\end{bmatrix}
\end{equation*}
and 
\begin{equation*}
P^{-1}=\begin{bmatrix} \frac{1}{2}-\frac{b}{c} & \frac{1}{4} & -\frac{1}{4} & -\frac{b}{c} \\\frac{1}{2}-\frac{a}{c} & -\frac{1}{4} & \frac{1}{4} & -\frac{a}{c}\\0 & \frac{1}{2} & \frac{1}{2} & 0\\ \frac{1}{c} & 0 & 0 & \frac{1}{c}\end{bmatrix}
\end{equation*} 
with
\begin{equation*}
a=\frac{\alpha\beta}{2}(C_{1,2}-C_{2,1})+\frac{\alpha^2}{2}C_{1,2}-\beta^2
\end{equation*}
\begin{equation*}
b=-\frac{\alpha\beta}{2}(C_{1,2}-C_{2,1})+\frac{\alpha^2}{2}C_{1,2}-\beta^2,
\end{equation*}
and $c=\alpha^2-4\beta^2$. Thus
\begin{equation}
\exp(\begin{bmatrix} B_{11} & G_1 \\F_1 & B_{15}\end{bmatrix}t)=Pe^{-\alpha t}\begin{bmatrix} e^{-2\beta t} &  &  &  \\ & e^{2\beta t} &  &  \\ &  & 1 & \\  &  &  & e^{\alpha t}\end{bmatrix}P^{-1}
\end{equation}

The nonzero elements (the 1st, 2nd, 9th, 10th elements, denoted as $d_1$, $d_2$, $d_9$, $d_{10}$) of the 1st row of $\exp(A_1t)$ are the 1st row of $\exp(\begin{bmatrix} B_{11} & G_1 \\F_1 & B_{15}\end{bmatrix}t)$. The nonzero elements (the 3rd, 4th, 11th, 12th elements, denoted as $e_3$, $e_4$, $e_{11}$, $e_{12}$) of the 3rd row of $\exp(A_2t)$ are the 1st row of $\exp(\begin{bmatrix} B_{22} & G_2 \\F_2 & B_{26}\end{bmatrix}t)$. The nonzero elements (the 5th, 6th, 13th, 14th elements, denoted as $f_5$, $f_6$, $f_{13}$, $f_{14}$) of the 5th row of $\exp(A_3t)$ are the 1st row of $\exp(\begin{bmatrix} B_{33} & G_3 \\F_3 & B_{37}\end{bmatrix}t)$.  The nonzero elements (the 7th, 8th, 15th, 16th elements, denoted as $g_7$, $g_8$, $g_{15}$, $g_{16}$) of the 5th row of $\exp(A_4t)$ are the 1st row of $\exp(\begin{bmatrix} B_{44} & G_4 \\F_4 & B_{48}\end{bmatrix}t)$. Thus
\begin{equation}
\begin{array}{lcl}
\vv{\rho}_t(1) & = & \vv{\rho}_0(1)d_1+\vv{\rho}_0(2)d_2+\vv{\rho}_0(9)d_9+\vv{\rho}_0(10)d_{10} 
\\[2ex]
& = & \rho_0(1,1)d_1+\rho_0(2,1)d_2+\rho_0(1,2)d_9+\rho_0(2,2)d_{10}
\\[2ex]
& = & \rho_0(1,1)d_1+\rho_0(2,2)d_{10}
\\[2ex]
& = & \displaystyle p \left(\frac{1-p_A}{2} \right) d_1 + p \left(\frac{1-p_A}{2} \right) d_{10}
\\[2ex]
& = & \displaystyle p(1-p_A) \left[ \frac{1}{4}e^{-\alpha t}(e^{-2\beta t}+e^{2\beta t})+d_{10} \right]
\end{array}
\end{equation}
where $d_{10}=\frac{\alpha^2C_{1,2}-2\beta^2}{c}-\frac{e^{-\alpha t}}{c}(ae^{2\beta t}+be^{-2\beta t})$. Similarly, we have 
\begin{equation}
\vv{\rho}_t(19)= p \cdot p_A \left[ \frac{1}{4}e^{-\alpha t}(e^{-2\beta t}+e^{2\beta t})+e_{12} \right]
\end{equation}
\begin{equation}
\vv{\rho}_t(37)= (1-p)(1-p_A) \left[ \frac{1}{4}e^{-\alpha t}(e^{-2\beta t}+e^{2\beta t})+f_{14} \right]
\end{equation}
\begin{equation}
\vv{\rho}_t(55)= p \cdot p_A \left[ \frac{1}{4}e^{-\alpha t}(e^{-2\beta t}+e^{2\beta t})+g_{16} \right]
\end{equation}
where 
\begin{equation*}
e_{12}=\frac{\alpha^2C_{3,4}-2\beta^2}{c}-\frac{e^{-\alpha t}}{c}(ae^{2\beta t}+be^{-2\beta t}),
\end{equation*}
\begin{equation*}
f_{14}=\frac{\alpha^2C_{5,6}-2\beta^2}{c}-\frac{e^{-\alpha t}}{c}(ae^{2\beta t}+be^{-2\beta t}),
\end{equation*}
\begin{equation*}
g_{16}=\frac{\alpha^2C_{7,8}-2\beta^2}{c}-\frac{e^{-\alpha t}}{c}(ae^{2\beta t}+be^{-2\beta t}).
\end{equation*}

Summing up $\vv{\rho}_t(1)$, $\vv{\rho}_t(19)$, $\vv{\rho}_t(37)$, and $\vv{\rho}_t(55)$ leads to the claimed $Pr(C)$.

\end{document}